\documentclass[aps,amsmath,amssymb,amsfonts,nofootinbib,superscriptaddress,tightenlines,12pt]{revtex4}

\usepackage{bm,graphicx,mathrsfs}

\usepackage{amsmath,bbm}
\usepackage{amsfonts,amssymb}
\usepackage{graphicx}
\usepackage{epsfig}
\usepackage{times}
\usepackage{verbatim}

\newenvironment{myproof}{{{\noindent \textsc{Proof By Contradiction}:}}}{\hfill$\square$}

\newtheorem{lemm}{Lemma}
\newtheorem{resul}{Result}

\newcommand{\be}{\begin{equation}}
\newcommand{\ee}{\end{equation}}
\newcommand{\bq}{\begin{eqnarray}}
\newcommand{\eq}{\end{eqnarray}}
\newcommand{\bea}{\begin{eqnarray}}
\newcommand{\eea}{\end{eqnarray}}
\newcommand{\ba}{\begin{align}}
\newcommand{\ea}{\end{align}}

\newcommand{\I}{\mathbb {I}}
\newcommand{\1}{\mathbbm{1}}

\newcommand{\ket}[1]{\left |  #1 \right\rangle}
\newcommand{\bra}[1]{\left \langle #1 \right |}

\newcommand{\proj}[1]{\ket{#1}\!\bra{#1}}

\newcommand{\tr}{\textrm{tr}}

\newcommand{\cT}{\mathcal{T}}

\newcommand{\rhobig}{\textrm{\large $\rho$}}
\newcommand{\sigmabig}{\textrm{\large $\sigma$}}

\newcommand{\tth}[0]{\textsuperscript{{th}}}

\newcommand{\qedsymbol}{\hfill $\square$ \vspace{2 mm}}

\usepackage{xcolor}

\DeclareMathAlphabet{\matheu}{U}{eus}{m}{n}
\usepackage{hyperref}
\hypersetup{
    bookmarksnumbered=true, 
    unicode=false, 
    pdfstartview={}, 
    pdftitle={}, 
    pdfauthor={}, 
    pdfsubject={}, 
    pdfcreator={}, 
    pdfproducer={}, 
    pdfkeywords={}, 
    pdfnewwindow=true, 
    colorlinks=true, 
    linkcolor=blue, 
    citecolor=blue, 
    filecolor=blue, 
    urlcolor=blue 
}

\newcommand{\eqn}[1]{\hyperref[eqn:#1]{Eq.~(\ref*{eqn:#1})}}
\newcommand{\eqrange}[2]{Eqs. (\ref{eqn:#1}--\ref{eqn:#2})}
\renewcommand{\sec}[1]{\hyperref[sec:#1]{Section~\ref*{sec:#1}}}
\newcommand{\thm}[1]{\hyperref[thm:#1]{Theorem~\ref*{thm:#1}}}
\newcommand{\lem}[1]{\hyperref[lem:#1]{Lemma~\ref*{lem:#1}}}
\newcommand{\prop}[1]{\hyperref[prop:#1]{Proposition~\ref*{prop:#1}}}
\newcommand{\cor}[1]{\hyperref[cor:#1]{Corollary~\ref*{cor:#1}}}
\newcommand{\fig}[1]{\hyperref[fig:#1]{Figure~\ref*{fig:#1}}}
\newcommand{\app}[1]{\hyperref[app:#1]{Appendix~\ref*{app:#1}}}
\newcommand{\defi}[1]{\hyperref[defi:#1]{Definition~\ref*{defi:#1}}}
\newcommand{\clm}[1]{\hyperref[clm:#1]{Claim~\ref*{clm:#1}}}
\newcommand{\res}[1]{\hyperref[res:#1]{Result~\ref*{res:#1}}}

\newcommand{\sop}[1]{{\mathcal #1}}

\def\qed{\leavevmode\unskip\penalty9999 \hbox{}\nobreak\hfill
     \quad\hbox{\leavevmode  \hbox to.77778em{%
               \hfil\vrule   \vbox to.675em%
               {\hrule width.6em\vfil\hrule}\vrule\hfil}}
     \par\vskip3pt}
    {\hspace*{\fill}$\Box$\vspace{1.5ex}\par}

\newcommand{\oZ}{\hat{S}}

\newcommand{\no}{\nonumber\\}

\usepackage{color}

\usepackage{fancyhdr} 
\begin{document}

\title{\sc{\Large A Quantum Version of Sch\"oning's Algorithm Applied to Quantum 2-SAT}}

\author{Edward Farhi}
\affiliation{Center for Theoretical Physics, Massachusetts Institute of Technology, Cambridge, MA 02139, USA}
\author{Shelby Kimmel}
\affiliation{Center for Theoretical Physics, Massachusetts Institute of Technology, Cambridge, MA 02139, USA}
\affiliation{Joint Center for Quantum Information and Computer Science (QuICS), University of Maryland,
College Park, MD 20742, USA}
\author{Kristan Temme}
\affiliation{Center for Theoretical Physics, Massachusetts Institute of Technology, Cambridge, MA 02139, USA}
\affiliation{Institute for Quantum Information and Matter, California Institute of Technology, Pasadena CA 91125, USA}
\affiliation{IBM TJ Watson Research Center, Yorktown Heights, NY 10598, USA}
\date{\today}

\begin{abstract}  
We study a quantum algorithm that consists of a
simple quantum Markov process, and we analyze its behavior on restricted
versions of Quantum 2-SAT.  We prove that the algorithm solves this
decision problem with
high probability for $n$
qubits, $L$ clauses, and promise gap $c$ in time ${\cal O}(n^2 L^2 c^{-2})$. If the
Hamiltonian is additionally polynomially gapped, our algorithm efficiently produces
a state that has high overlap with the satisfying subspace. The
Markov process we study is a quantum
analogue of Sch\"oning's probabilistic algorithm for $k$-SAT.
\end{abstract} \maketitle

\thispagestyle{fancy} 
\renewcommand{\headrulewidth}{0pt}
\rhead{MIT-CTP/4789}
\section{Introduction}\label{sec:intro}

For the $n$-bit classical constraint satisfaction problem $k$-SAT,
several algorithms beat the exhaustive search
runtime bound of $2^n$. They provide a runtime with a mildly
exponential scaling, $O(r^n)$ with $r < 2$. One such algorithm
is Sch\"oning's probabilistic algorithm that finds a solution of 3-SAT
in time $O(1.334^n)$~\cite{Schoennig}. The algorithm
works by exploring the solution space using a simple Markov process.
Although variants of the algorithm had been known for some time
\cite{minton1992minimizing,papadimitriou}, Sch\"oning was the first
to prove the runtime bound for $k\geq 3$. For 2-SAT,
Papadimitriou earlier introduced a variant of this algorithm that
finds a satisfying assignment (if there is one)
in time $O(n^2)$
~\cite{papadimitriou}. While linear-time 2-SAT algorithms 
exist~\cite{even1975complexity,aspvall1979linear}, Papidimitriou's
algorithm is admired for its simplicity. 

Quantum $k$-SAT is the quantum generalization of the classical
$k$-SAT problem. Analogously to classical $k$-SAT, 
Quantum 3-SAT is $\mbox{QMA}_1$-complete
\cite{Gosset}, while Quantum 2-SAT can be solved in polynomial
time \cite{B06}. Interestingly, existing algorithms for Quantum 
2-SAT have paralleled algorithms for classical 2-SAT: 
Bravyi's original algorithm for Quantum 2-SAT
is similar to Krom's algorithm for classical 
2-SAT~\cite{krom1967decision}
and uses inference rules; and two recent linear-time 
algorithms for Quantum 2-SAT \cite{dBG15,AS+15} use
ideas from linear-time classical 2-SAT algorithms
\cite{even1975complexity,aspvall1979linear}. 

In this work, we describe an algorithm that is a quantum
analogue of Papidimitriou's classical algorithm and analyze its
behavior on restricted versions of Quantum 2-SAT. Like the classical
algorithm, our quantum version consists of repeated 
applications of a simple (quantum) Markov process. 
As with the recent linear-time Quantum 2-SAT algorithms, we apply
tools and intuition from the classical algorithm to analyze
the quantum version. However, our algorithm is a quantum
algorithm; past algorithms for Quantum 2-SAT have been
classical. Since Sch\"oning showed that 
the classical version of this algorithm performs well for classical
$k$-SAT with $k>2$, there is hope that the quantum version will have 
success on Quantum $k$-SAT with $k>2$. Therefore, we think understanding
this quantum Markov process in the case of $k=2$ is of value.

Papidimitriou's classical algorithm for 2-SAT takes
as input the number of bits $n$, 
a set of clauses $\sop I$, and a real parameter $b>0$, where
$b$ is chosen
depending on the desired probability of success. Then
the algorithm is as follows: \\

{\it Classical Algorithm}$(n, \sop I, b)$
\begin{itemize}
\item Pick a string $s$ uniformly at random from $\{0,1\}^n$.
\item  Repeat $bn^2$ times: 
\begin{itemize}
\item If there exist clauses in $\sop I$ that are not satisfied on $s$, randomly choose one of the unsatisfied
 clauses, and then randomly choose one of the bits in that clause. Flip the value of that bit
 and rename $s$ to be the new string with the flipped bit.
 \item If $s$ satisfies all clauses, return $s$ and terminate.
\end{itemize} 
\item If $s$ does not satisfy all clauses, return ``No satisfying string found.''
\end{itemize}
If there is no satisfying assignment, the algorithm will always return ``No satisfying string found.''
 If a satisfying string exists, this algorithm will return a satisfying assignment with probability $p$,
where $(1-p)\propto b^{-1}$.

 The quantum algorithm that we consider is the natural generalization
 of this procedure to the quantum domain for the problem Quantum
 $k$-SAT, which is the natural generalization of Classical $k$-SAT to the
 quantum domain. We now give the definition of Quantum $k$-SAT on $n$ qubits as it
 was introduced by Bravyi (altered to include only rank-1 projectors)
  \cite{B06}:

\vspace{5mm}\noindent{\bf Definition [Quantum k-SAT]} {\it
Let  $c = \Omega(n^{-g})$  with $g$ a positive constant. Given a set
of $L$ rank one projectors (called ``clauses'') ${\Phi}_\alpha =
\proj{\phi_\alpha}$ each supported on $k$ out of $n$ qubits, define 
\begin{align}
H=\sum_{\alpha=1}^L\Phi_\alpha.
\end{align}
One must
decide between the following two cases:
\begin{enumerate}
 \item The YES instance: There exists an $n$-qubit state $\rhobig$ that satisfies $\tr[H\rhobig] = 0$.
\item The NO instance: For any $n$-qubit state $\rhobig$, we have that  $\tr[H\rhobig] \geq c$.
\end{enumerate} }\vspace{5mm}


We now give a quantum algorithm for Quantum $k$-SAT on $n$ qubits, but
in this paper we focus on $k=2$. The quantum algorithm takes as input
the number of qubits  $n$, a set of $L$ clauses $\sop
I=\{\Phi_\alpha\}$, and two positive integers $N$ and $T$, where
$N\leq T$. $N$ and $T$ are chosen based on the desired probability of
success.  The clauses can be given either via a classical description,
or operationally, as measurement projectors. Then the algorithm is as
follows:\\

\vspace{.1cm}
{\it Quantum Algorithm}$(n,\sop I,N,T)$
\begin{itemize}
\item Initialize the system in the maximally mixed state of $n$ qubits.
\item Initialize a counter $N_0$ to equal 0. 
\item Repeat $T$ times:
\begin{itemize}
\item Choose $\alpha$ uniformly at random from $\{1,\dots,L\}$, and
measure $\Phi_\alpha$. If outcome 1 is measured, choose one of
the qubits in the support of ${\Phi}_\alpha$ at random and apply a
Haar random unitary to that qubit. If outcome 0 is measured, set
$N_0=N_0+1$.
\end{itemize} 
\item If $N_0\geq N$ decide you are in a YES instance. Otherwise,
decide NO. 
\end{itemize}

One might expect that an algorithm for Quantum $k$-SAT first prepares
a low energy state, and then estimates the energy of the
state using, for example, phase estimation. In our work we use the repeated
measurements of clauses to fulfill both roles. We prepare the low energy state by repeatedly
measuring clauses and applying random unitaries if the clauses are
unsatisfied. We test whether the state has low energy by
tracking the number of satisfied outcomes. We will show that if, over repeated measurements, most of the outcomes
are satisfied, then we have a low energy state.

Variants of this algorithm have been analyzed previously in different
contexts. A similar algorithm was proposed to prepare graph states and
Matrix Product States dissipatively \cite{Frank}, and a variant was
used as a tool for the constructive proof of a quantum local Lov\'asz
lemma for commuting projectors \cite{AS13,SCV13}.

Given a YES instance of Quantum $2$-SAT, since Quantum 2-SAT is in $P$, 
one might expect that the
{\it Quantum Algorithm} will converge to a satisfying state in
polynomial time. We show that this is
indeed the case, at least for a restricted set of clauses. Chen et al. 
\cite{CCD+11} showed that for every YES instance of Quantum 2-SAT, there
is always a satisfying assignment that  is a product of single- and two-qubit states.
In fact, with the restricted clause set that we consider, there
will be a satisfying single-qubit product state of the form:
\begin{align}
\ket{\psi_1}_1\otimes \cdots\otimes\ket{\psi_n}_n
\end{align}
where the ket $\ket{\cdot}_i$ denotes the state
of the $i\tth$ qubit. For ease of notation, for YES instances, we use the following basis:
\begin{align}
\ket{0}_i=\ket{\psi_i}_i.
\end{align}
Hence, for the rest of this paper, $\ket{0}^{\otimes n}$ does not refer to the standard basis
state, but to an unknown product state that satisfies all clauses of a 
Quantum 2-SAT instance. 
In the basis where $\ket{0}^{\otimes n}$ is a satisfying state,
 all of the clauses are of the form \\

\noindent{\bf General Clauses:} 
\be \label{eqn:general} 
\Phi_\alpha=\proj{\phi_\alpha}, \hspace{.2cm} \textrm{ with } 
\ket{\phi_\alpha}=a_\alpha\ket{01}_{i,j}+
b_\alpha\ket{10}_{i,j}+c_\alpha\ket{11}_{i,j}, 
\ee

\noindent where  $i,j$ label the two qubits in the clause $\Phi_\alpha$.
For reasons that we will discuss later, we can only prove that the
{\it Quantum Algorithm} succeeds in polynomial time if in the YES
instance the clauses are restricted to have $c_\alpha=0$.
In the NO case, the clauses have no restrictions.  We call this
problem {\bf Restricted Quantum 2-SAT}, and we show that the {\it
Quantum Algorithm} can succeed in this setting when $T=O(L^4n^2/c^2)$.
This restriction can be somewhat relaxed, and in \app{extended-clause}, we show
that the algorithm succeeds in polynomial time if in the YES instance every clause satisfies
either $c_\alpha=0$ or $a_\alpha=b_\alpha=0$. So for now we work
with\\

\noindent{\bf Restricted Clauses:} 
\be \label{eqn:restrict} 
\Phi_\alpha=\proj{\phi_\alpha}, \hspace{.2cm} \textrm{ with } 
\ket{\phi_\alpha}=a_\alpha\ket{01}_{i,j}+b_\alpha\ket{10}_{i,j}.
\ee

\noindent Note that $\ket{0}^{\otimes n}$ and $\ket{1}^{\otimes n}$ are both
satisfying states with the restricted clause set. 


In addition to solving Restricted Quantum 2-SAT, in the YES case 
the {\it Quantum Algorithm} produces a state
that has high overlap with a satisfying assignment. In this setting,
the smallest eigenvalue of $H$ is  0, and we call $\epsilon$ the size of
the smallest non-zero eigenvalue of $H$. We show that after running
the {\it Quantum Algorithm} for $T=O(n^2L/\epsilon)$ steps, the
resultant state will have large overlap with a state $\rhobig$ that has
$\tr[H\rhobig]=0$. 

The {\it{Quantum Algorithm}} may solve arbitrary Quantum 2-SAT
instances in polynomial time, but our analysis can only show that it
succeeds in polynomial time on Restricted Quantum 2-SAT. On the other
hand, Bravyi's algorithm and recent linear-time
quantum algorithms \cite{AS+15,dBG15} give procedures for  deciding all Quantum
2-SAT instances in polynomial time, but are classical algorithms.
Our algorithm is a quantum algorithm, so our analysis techniques may
be of broader interest. In particular, our approach may have
applications to Quantum $k$-SAT for $k>2.$


\section{Analysis of the Quantum Algorithm for Restricted Quantum 2-SAT}\label{sec:analysis}

On a YES instance, the {\it Quantum Algorithm} can be viewed as a quantum Markov process that converges
to a quantum state that is annihilated by all the clauses. A quantum
Markov process is described by a completely positive trace preserving
(CPTP)  map
\cite{davies1976quantum}. Call $\rhobig_t$ the state of the system
at time $t$. The CPTP map $\cT$
describes the update of $\rhobig_t$ at each step of the chain, so
$\rhobig_{t+1} = \cT(\rhobig_t)$. 

Call $\cT_{\alpha}$ the map that describes the procedure of
checking whether clause $\Phi_{\alpha}$ is satisfied, and if it is
not satisfied, applying a random unitary to one of the qubits in the
support of $\Phi_{\alpha}$. Let $i$ and $j$ be the two qubits associated
with clause $\Phi_\alpha.$ Then
\begin{align}
 \cT_{\alpha}(\rhobig)  =\left(\1{-}\Phi_{\alpha} \right) \rhobig \left(\1 {-} \Phi_{\alpha} \right) + 
\textrm{\small $\frac{1}{2}$}\Lambda_i(\Phi_{\alpha} \rhobig \Phi_{\alpha})+\textrm{\small $\frac{1}{2}$}\Lambda_j(\Phi_{\alpha} \rhobig \Phi_{\alpha})
\end{align}
where $\Lambda_i$ is the unitary twirl map acting on qubit $i$:
\begin{align}
\Lambda_i(\rhobig)=\int d[U_i] U_{i}\rhobig U_{i}^\dagger =\frac{\1_i}{2}\otimes\textrm{tr}_i\left[\rhobig\right],
\end{align}
and $d[U_i]$ is the Haar measure.
At each time step, we choose $\alpha$ from $\{1,\dots,L\}$ uniformly
and random and apply the map $\cT_{\alpha}$. This corresponds to the
CPTP update map
\begin{align}
\cT(\rhobig)  =&\frac{1}{L} \sum_{\alpha=1}^L \cT_{\alpha}(\rhobig) \label{eqn:cTdef}.
\end{align}

During the measurement step, when $\alpha$ is chosen uniformly at random
and one measures $\Phi_\alpha$, the probability of obtaining outcome
$1$ at time $t$ is 
\begin{align}
\frac{1}{L}\sum_\alpha\tr[\Phi_\alpha\rho_t]=\frac{1}{L}\tr[H\rho_t].
\end{align}


\subsection{Expectation of Total Spin}
In analyzing the classical algorithm, Papadimitriou and Sch\"oning
kept track of the Hamming distance between the current string
and the satisfying assignment. Inspired by this idea, we find
it useful to analyze the expectation value of
$\oZ$ and $\oZ^2$, where $\oZ$ is twice the total spin:

\begin{align}
\oZ=\sum_{i=1}^n\sigmabig^z_i \hspace{.5cm}\textrm{and}\hspace{.5cm}\oZ^2=\sum_{i,j=1}^n\sigmabig^z_i\sigmabig^z_j.
\end{align}

\noindent Note that $\oZ$ is closely related to the quantum Hamming weight
operator $\sum_{i=1}^n\frac{1}{2}(1-\sigmabig^z_i)$.

We show that with the restricted clause set, the expectation value of $\oZ$ is constant under the action
of $\cT$, whereas the expectation value of $\oZ^2$ can not decrease under the action of
$\cT$.

\begin{lemm}\label{lem:slem}
Given a set of restricted clauses $\{\Phi_1,\dots,\Phi_L\}$ (i.e. all of the form of \eqn{restrict}), with 
$\cT$ defined as in \eqn{cTdef}, then
\begin{align}
{\rm tr}[\oZ\cT(\rhobig)]-{\rm tr}[\oZ\rhobig]&=0\\
{\rm tr}[\oZ^2\cT(\rhobig)]-{\rm tr}[\oZ^2\rhobig]&=
\frac{2}{L} \sum_{\alpha}{\rm tr}[\Phi_{\alpha}\rhobig]
\geq 0.\label{eqn:Zsqrd_change}
\end{align}
\end{lemm}

\begin{proof}
Let $\cT^\dagger$ be the dual of $\cT$, so that
\begin{align}
 \tr[\oZ\cT(\rhobig)]= \tr[\cT^\dagger(\oZ)\rhobig]\hspace{.5cm}\textrm{and}\hspace{.5cm}
  \tr[\oZ^2\cT(\rhobig)]= \tr[\cT^\dagger(\oZ^2)\rhobig].
 \end{align} 
 First consider
 \begin{align}
 \cT^\dagger_\alpha(\oZ)=(1{-}\Phi_{\alpha})\oZ(1{-}\Phi_{\alpha})
 +\textrm{\small $\frac{1}{2}$}\Phi_{\alpha}\Lambda_i(\oZ)\Phi_{\alpha}
 +\textrm{\small $\frac{1}{2}$}\Phi_{\alpha}\Lambda_j(\oZ)\Phi_{\alpha},
 \end{align}
where $i,j$ are the two qubits where $\Phi_{\alpha}$ acts. Note
that $\oZ-\sigmabig_{i}^z-\sigmabig_{j}^z$ is invariant
under the action of $\cT^\dagger_\alpha$, so
 \begin{align} \label{eqn:expand-oZ}
\cT_\alpha^\dagger(\oZ) 
=&\oZ-\sigmabig_{i}^z-\sigmabig_{j}^z
+\left(\1 {-} \Phi_{\alpha} \right)(\sigmabig_{i}^z+\sigmabig_{j}^z) \left(\1 {-} \Phi_{\alpha} \right) \no 
+& \textrm{\small $\frac{1}{2}$}\Phi_{\alpha}\Lambda_{i}(\sigmabig_{i}^z+\sigmabig_{j}^z) \Phi_{\alpha} + 
 \textrm{\small $\frac{1}{2}$}\Phi_{\alpha}\Lambda_{j}(\sigmabig_{i}^z+\sigmabig_{j}^z) \Phi_{\alpha}.
\end{align}

\noindent Due to the special properties of the restricted clauses, c.f.
\eqn{restrict}, we have
\begin{align}\label{eqn:typeIprops}
\Phi_{\alpha} (\sigmabig_{i}^z+\sigmabig_{j}^z) =(\sigmabig_{i}^z+\sigmabig_{j}^z)\Phi_{\alpha} =0,
\end{align}
for all $\alpha$, which together with $\Lambda_i(\sigmabig_i^z)=0$ and $\Lambda_i(\sigmabig_j^z)=\sigmabig_j^z$
for $i\neq j$ gives
\begin{align}
\cT_\alpha^\dagger(\oZ)=\oZ.
\end{align}
This implies
\begin{align}
\cT^\dagger(\oZ)=\oZ,
\end{align}
so we see that the expectation value of $\oZ$ is unchanged by the action of $\cT$ on a state:
\begin{align}
\tr[\oZ\cT(\rhobig)]=\tr[\oZ\rhobig].
\end{align}

The expectation value of $\oZ^2$ does change under the action of $\cT$.  $\Phi_\alpha$ acts only on qubits
$i$ and $j$, so accordingly we break up $\oZ^2$ as
\begin{align}\label{eqn:S2divide}
\oZ^2=&{\bigg[}\oZ^2-2\sigmabig_i^z\sigmabig_j^z-
2\sum_{k\neq i,j}\sigmabig_k^z(\sigmabig_i^z+\sigmabig_j^z)\bigg]\no
&+\bigg[2\sigmabig_i^z\sigmabig_j^z+
2\sum_{k\neq i,j}\sigmabig_k^z(\sigmabig_i^z+\sigmabig_j^z)\bigg].
\end{align}
$\cT_\alpha^\dagger$ leaves the first term unchanged. Now 
 \begin{align}
 \cT^\dagger_\alpha(\sigmabig_i^z\sigmabig_j^z)=(1{-}\Phi_{\alpha})\sigmabig_i^z\sigmabig_j^z(1{-}\Phi_{\alpha})
 +\textrm{\small $\frac{1}{2}$}\Phi_{\alpha}\Lambda_i(\sigmabig_i^z\sigmabig_j^z)\Phi_{\alpha}
 +\textrm{\small $\frac{1}{2}$}\Phi_{\alpha}\Lambda_j(\sigmabig_i^z\sigmabig_j^z)\Phi_{\alpha}.
 \end{align}
 Because of the special properties of the clauses, c.f. \eqn{restrict}, we have
  \begin{align}\label{eqn:typeIprops2}
\Phi_{\alpha} \sigmabig_{i}^z\sigmabig_{j}^z=
\sigmabig_{i}^z\sigmabig_{j}^z\Phi_{\alpha}=-\Phi_{\alpha}.
\end{align}

\noindent Using \eqn{typeIprops} and that $\Lambda_i(\sigmabig_i^z)=0$, we have
\begin{align}
 \cT^\dagger_\alpha(\sigmabig_i^z\sigmabig_j^z)=
 \sigmabig_i^z\sigmabig_j^z+\Phi_\alpha.
\end{align}

\noindent Now notice
\begin{align}\label{eqn:S2final}
\cT_\alpha^\dagger(\sigmabig_k^z(\sigmabig_i^z+\sigmabig_j^z))
=& (1{-}\Phi_{\alpha})\sigmabig_k^z(\sigmabig_i^z+\sigmabig_j^z)(1{-}\Phi_{\alpha})\no
 &+\textrm{\small $\frac{1}{2}$}\Phi_{\alpha}\Lambda_i(\sigmabig_k^z(\sigmabig_i^z+\sigmabig_j^z))\Phi_{\alpha}\no
 &+\textrm{\small $\frac{1}{2}$}\Phi_{\alpha}\Lambda_j(\sigmabig_k^z(\sigmabig_i^z+\sigmabig_j^z))\Phi_{\alpha}\no
 =&\sigmabig_k^z(\sigmabig_i^z+\sigmabig_j^z).
\end{align}
where we have again used \eqn{typeIprops}. Putting the pieces together gives
\be \cT_\alpha^\dagger(\oZ^2)=\oZ^2+2\Phi_{\alpha}. 
\ee

\noindent The change in the expectation value of $\oZ^2$ after the action of $\cT$ is thus
  
\begin{align}
\tr[\oZ^2\cT(\rhobig)]-\tr[\oZ^2\rhobig]=
\frac{2}{L} \sum_{\alpha}\tr[\Phi_{\alpha}\rhobig]
\geq 0.\label{eqn:Z2_change}
\end{align}
\qedsymbol
\end{proof}

\subsection{Runtime of the {\it Quantum Algorithm} to Decide Restricted Quantum 2-SAT}\label{sec:decide}

The {\it Quantum Algorithm} decides between YES and NO cases based on
the number of $0$-valued outcomes, i.e. satisfied projectors, obtained
during the algorithm. The
probability of getting a $0$-outcome at step $t$ is
\begin{align}
1-\frac{1}{L}\tr[H\rhobig_t],
\end{align}
and so depends on the expectation value of $H.$ \eqn{Z2_change} allows
us to relate the expectation value of $H$ to the expectation 
value of $\oZ^2$. While the expectation value of $H$ is not necessarily monotonic
over the course of the algorithm, the expectation value of $\oZ^2$ is
monotonic (by \lem{slem}) and is also bounded, since the maximum
eigenvalue of $\oZ^2$ on $n$ qubits is $n^2$. We use these properties of $\oZ^2$
to track the expectation value of $H$ over the course of the algorithm, and hence
to track the expected number of $0$-valued outcomes.

We analyze the YES and NO cases seperately.
\begin{resul}\label{res:T1overlap}
Suppose we have a {\em YES} case of Restricted Quantum 2-SAT, and we run the {\em Quantum Algorithm}
for time
\begin{align}
T&=\frac{f^2L^2n^2}{2},\label{eqn:Tdef}
\end{align}
where
\begin{align}
f=\max\left\{\frac{7}{c},1\right\},
\end{align}
then we have at least a $2/3$ probability of observing at least $N$ measurement outcomes
with value $0$ over the course of the algorithm, where
\begin{align}
N=T\left(\frac{fL-1}{fL}\right)^3-fLn.\label{eqn:Neqn}
\end{align}
\end{resul}
The choice of $f=\max\left\{\frac{7}{c},1\right\}$ is not used in
this proof, but is rather important for the soundness analysis. We include
it here for concreteness.\\

\begin{proof}
We start by using \lem{slem} to bound the expectation value of $H$
over the course of the algorithm. $0\leq\tr[\oZ^2\rhobig ]\leq n^2$
for any state $\rhobig$ on $n$ qubits and so for any $T$
\begin{align}
n^2&\geq\tr[\oZ^2\rhobig_T ]-\tr[\oZ^2\rhobig_0] \no
&=\sum_{t=0}^{T-1}\left( \tr[\oZ^2\rhobig_{t+1}]-\tr[\oZ^2\rhobig_{t}] \right)\no
&= \frac{2}{L}\sum_{t=0}^{T-1}\sum_\alpha\tr[\Phi_{\alpha}\rhobig_t]\nonumber\\
&= \frac{2}{L}\sum_{t=0}^{T-1} \tr[ H \rhobig_t].\label{eqn:sumboth1}
\end{align}

Let $\Pi_{f}$ be the projector onto the eigenstates of $H$ with eigenvalue less than
$1/f$. We define
\begin{align}
\textrm{\large $p$}_{t,f}=\tr[\Pi_{f}\rhobig_t].
\end{align}

Inserting the projector $\I-\Pi_{f}$ into the last line of \eqn{sumboth1},
we have
\begin{align}
n^2&\geq\frac{2}{L} \sum_{t=0}^{T-1} \tr[ H(\I-\Pi_{f})\rhobig_t]\no
&\geq \frac{2}{fL}\sum_{t=0}^{T-1}  (1-\textrm{\large $p$}_{t,f})
\end{align}
where we used that $\rhobig_t$ has probability $1-\textrm{\large $p$}_{t,f}$ of being
in the subspace $\I-\Pi_{f}$, and states in this subspace have
expectation value of $H$ at least $1/f$. Rearranging terms gives
\begin{align}
\sum_{t=0}^{T-1}\textrm{\large $p$}_{t,f}\geq T-\frac{fLn^2}{2},
\end{align}
and using \eqn{Tdef} gives
\begin{align}
\sum_{t=0}^{T-1}\textrm{\large $p$}_{t,f}\geq \frac{fL-1}{fL}T.
\end{align}

By the pigeon hole principle, there is a set of 
times $\mathbb T$ such that the following are true:
\begin{align}
\textrm{\large $p$}_{t,f}&\geq \frac{fL-1}{fL}\textrm{ }\text{ for }t\in \mathbb T, \textrm{ }\text{ and }\label{eqn:pbound}\\
|\mathbb T|&\geq \frac{fL-1}{fL}T.\label{eqn:Tbound}
\end{align}

At any time $t$, the probability of obtaining outcome 0 is
\begin{align}
1-\frac{1}{L}\sum_{\alpha=1}^L\tr[\Phi_\alpha\rhobig_t]&=
1-\frac{1}{L}\tr[H((\I-\Pi_{f})+\Pi_{f})\rhobig_t].
\end{align}
Since $H$ is a sum of $L$ projectors, its eigenvalues are at most $L$, so we have
\begin{align}
\tr[H(\I-\Pi_{f})\rhobig_t]\leq L(1-\textrm{\large $p$}_{t,f}).
\end{align} 
$\Pi_{f}$ projects onto states with eigenvalue less than $1/f$, so
\begin{align}
\tr[H\Pi_{f}\rhobig_t]<\frac{1}{f}\textrm{\large $p$}_{t,f}.
\end{align} 
Plugging these in gives
\begin{align}
1-\frac{1}{L}\sum_{\alpha=1}^L\tr[\Phi_\alpha\rhobig_t]&
\geq\frac{fL-1}{fL} \textrm{\large $p$}_{t,f}.
\end{align}

Now assume $t\in\mathbb T,$ so \eqn{pbound} holds. Then 
we have for these times that the probability of obtaining
outcome 0 is
\begin{align}
1-\frac{1}{L}\sum_{\alpha=1}^L\tr[\Phi_\alpha\rhobig_t]&\geq
\left(\frac{fL-1}{fL}\right)^2
\end{align}

Since we want a large number of $0$-outcomes over the course of the
algorithm, we will assume a worst case scenario such that the
probability of outcome 0 for all times $t\in \mathbb T$ is 
\begin{align}
\textrm{\large $p$}_{\rm{worst}}=\left(\frac{fL-1}{fL}\right)^2.\label{eqn:pworst}
\end{align}
In this case, the distribution of $0$-outcomes for times $t\in \mathbb T$
is given by a binomial distribution. We can use bounds on the binomial
cumulative distribution function to bound the number of $0$-outcomes
in this worst case scenario. Let $G$ be the probability that less than
$N$ outcomes are 0 over $|\mathbb T|$ times, where $\textrm{\large $p$}_{\rm{worst}}$ is the
probability of obtaining outcome $0$ at any time. Using Hoeffding's
bound, we have that
\begin{align}
G\leq \exp \left[ \frac{-2(|\mathbb T|\textrm{\large $p$}_{\rm{worst}}-N)^2}{|\mathbb T|}\right], \label{eqn:Geqn}
\end{align}
as long as $|\mathbb T|\textrm{\large $p$}_{\rm{worst}}\geq N$.
Using \eqn{Tbound} and \eqn{pworst}, we have
\begin{align}
|\mathbb T|\textrm{\large $p$}_{\rm{worst}}\geq\left(\frac{fL-1}{fL}\right)^3T.
\end{align}
Using \eqn{Neqn},
we see that 
\begin{align}
|\mathbb T|\textrm{\large $p$}_{\rm{worst}}- N\geq fLn,
\end{align}
so the numerator of the exponent in \eqn{Geqn} satisfies 
\begin{align}
2(|\mathbb T|\textrm{\large $p$}_{\rm{worst}}-N)^2\geq 2f^2L^2n^2.
\end{align}
Finally, the denominator of the exponent in \eqn{Geqn} satisfies
\begin{align}
\mathbb T&\leq T \no
&= \frac{f^2L^2n^2}{2},
\end{align}
so we have 
\begin{align}
G\leq \exp \left[ -4\right]\leq 1/3. \label{Geqn}
\end{align}
Thus with probability at least 2/3, we expect to see at least 
$N$ outcomes with value $0$ for times $t\in \mathbb T$.
Considering times $t$ with $1\leq t\leq T$ rather than only times $t\in \mathbb T$ only
gives more opportunities for $0$-outcomes, so we have probability of at least 
$2/3$ of seeing $N$ outcomes with value $0$ when the algorithm is run
for time $T$.
\qedsymbol
\end{proof}

Now we prove an analogous result in the NO case:
\begin{resul}\label{res:T2overlap}
Recall that in the NO case, the size of the smallest eigenvalue of $H$
is promised to be $c$. If we run the
algorithm for time
\begin{align}
T&=\frac{f^2L^2n^2}{2},\label{eqn:Tdef2}
\end{align}
and choose
\begin{align}
f=\max\left\{\frac{7}{c},1\right\},
\end{align}
then we have at most a $1/3$ probability of observing more than $N$ measurement outcomes
with value $0$ over the course of the algorithm, where, as in \res{T1overlap},
\begin{align}
N=T\left(\frac{fL-1}{fL}\right)^3-fLn.\label{eqn:Nbound2}
\end{align}
\end{resul}

\begin{proof}
We show that if we have a NO case, we are unlikely to have more than 
$N$ measurements with outcome 0 over the course of the $T$ applications of $\cT$. 
In the NO case, the probability of obtaining outcome 0 at time $t$ is
\begin{align}
1-\frac{1}{L}\sum_{\alpha=1}^L\tr[\Phi_\alpha\rhobig_t]\leq 1-\frac{c}{L}.
\end{align}
The worst case is when for all times $t$, the probability of obtaining outcome 0 is 
\begin{align}
\textrm{\large $q$}_{\rm{worst}}=1-\frac{c}{L}.\label{eqn:qworst}
\end{align}
This worst case scenario corresponds to a binomial distribution.
We use bounds on the binomial distribution to bound the probability 
of at least $N$ outcomes with value $0$.
Let $\sop G$ be the probability of getting at least $N$ outcomes
with value $0$ over $T$ steps, where $\textrm{\large $q$}_{\rm{worst}}$ is the probability of obtaining outcome $0$
at any step. 
Applying Hoeffding's bound to the binomial distribution, we have
\begin{align}
\sop G\leq \exp\left[\frac{-2(N-T\textrm{\large $q$}_{\rm{worst}})^2}{T}\right]\label{eqn:Hoeff}
\end{align}
as long as $N\geq T\textrm{\large $q$}_{\rm{worst}}$. We now show that $\sop G$ is small. 

We first analyze the term $N-T\textrm{\large $q$}_{\rm{worst}}$ from \eqn{Hoeff}. 
We have, using \eqn{Nbound2}, 
\begin{align}
N-T\textrm{\large $q$}_{\rm{worst}}&=T\left(\frac{fL-1}{fL}\right)^3-fLn-T\left(1-\frac{c}{L}\right).
\end{align}
Since $fL\geq 1$, we have
\begin{align}
N-T\textrm{\large $q$}_{\rm{worst}}&\geq T\left(1-\frac{3}{fL}\right)-fLn-T\left(1-\frac{c}{L}\right)\no
&=\frac{1}{2}cf^2Ln^2-\frac{3}{2}fLn^2-fLn\no
&\geq fLn^2\left(\frac{cf}{2}-\frac{5}{2}\right)
\end{align}
where in the second to last line we used \eqn{Tdef2}, and in the last
line we used that $n\geq 1.$ Setting $f=\max\{7/c,1\},$ we have
\begin{align}
N-T\textrm{\large $q$}_{\rm{worst}}&\geq fLn^2,
\end{align}
where the maximum over the two terms is used to ensure $f\geq 1.$
Then the numerator in \eqn{Hoeff} satisfies
\begin{align}
2(N-T\textrm{\large $q$}_{\rm{worst}})^2\geq 4Tn^2.
\end{align}
Plugging into \eqn{Hoeff} we have
\begin{align}
\sop G\leq \exp[-4n^2]\leq 1/3.
\end{align}
Therefore, the probability of getting at least $N$ outcomes with value
0 is less than $1/3$.
\qedsymbol
\end{proof}

Combining \res{T1overlap} and \res{T2overlap}, to solve Restricted
Quantum $2$-SAT, we set 
\begin{align}
f=\max\left\{\frac{7}{c},1\right\}
\end{align} 
and run the algorithm for time
\begin{align}
T&=\frac{f^2L^2n^2}{2}.
\end{align}
We count the number of 0-outcomes over the course of the algorithm, and check whether 
this is greater than
\begin{align}
N=T\left(\frac{fL-1}{fL}\right)^3-fLn.
\end{align}
We have shown that for a YES instance, there is at least a 2/3 probability of observing 
at least $N$
outcomes with value $0$, but for a NO instance, there is at most a 1/3
probability of doing so.

\subsection{Runtime to Produce a Ground State}

Suppose we have a Hamiltonian with restricted clauses that is additionally
polynomially gapped. In other words,
the smallest non-zero eigenvalue of the Hamiltonian has size
$\Omega(1/\textrm{poly}(n))$. Then we show that repeatedly
applying the map $\sop T$ produces a state that has large overlap with
the ground subspace.


\begin{resul}\label{res:T3overlap}
Given clauses $\{\Phi_{\alpha}\}$ where $\Phi_{\alpha}
= \proj{\phi_\alpha}$ are restricted as in \eqn{restrict}, and
$\epsilon$ is the size of the smallest non-zero eigenvalue of
$H=\sum_\alpha\Phi_{\alpha}$, then for
$T\geq\frac{n^2L}{2(1-p)\epsilon}$, $\rhobig_{T}=\cT^{T}(\rhobig_0)$
has a fidelity $\tr[ \Pi_0\rhobig_{T}]$ with the ground state subspace that is greater than
$p$.
\end{resul}

\begin{myproof}

Let $\Pi_0$ be
the projector onto the satisfying subspace:
\begin{align}
\tr[H\Pi_0]=0.
\end{align}
We first show that $\Pi_0$ is a fixed point of the map $\cT$, so once part of the state
is in this subspace, it stays there. That is,
\begin{align}\label{eqn:Pi0inc}
\tr[\Pi_0\rhobig_{t+1}]-\tr[\Pi_0\rhobig_t]=
&\frac{1}{L}\sum_\alpha
\tr\Big[\Pi_0\left(\1 {-} \Phi_{\alpha} \right) \rhobig_t \left(\1 {-} \Phi_{\alpha} \right)\no 
& + \tfrac{1}{2}\Pi_0 \Lambda_i(\Phi_{\alpha} \rhobig_t \Phi_{\alpha})+
\tfrac{1}{2}\Pi_0 \Lambda_j(\Phi_{\alpha} \rhobig_t \Phi_{\alpha})\Big]-\tr[\Pi_0\rhobig_t]\no
=&\frac{1}{2L}\sum_\alpha\tr[\Pi_0\left(\Lambda_i(\Phi_{\alpha} \rhobig_t \Phi_{\alpha})+
\Lambda_j(\Phi_{\alpha} \rhobig_t \Phi_{\alpha})\right)]\no
\geq& 0,
\end{align}
since $\tr[\Pi\rhobig]\geq 0$ for any projector $\Pi$ and any state $\rhobig$. 

Suppose $\tr[\Pi_0\rhobig_T]<p$. From \eqn{Pi0inc},
$\tr[\Pi_0\rhobig_t]$ can not decrease with increasing $t$. So for all
$t\leq T$,
\begin{align}
\tr[\Pi_0\rhobig_t]< p
\end{align} or equivalently, 
\begin{align}\label{eqn:pioverlap}
\tr[(\I-\Pi_0)\rhobig_t]>1-p. 
\end{align}
Given that  the spectral gap of $H$ is $\epsilon$, we have 
\begin{align}\label{eqn:Hoverlap}
\tr[H \rhobig_t]>\epsilon\tr[(\I-\Pi_0)\rhobig_t].
\end{align}
Combining \eqrange{pioverlap}{Hoverlap} gives
\begin{align}
\tr[ H \rhobig_t]>\epsilon (1-p)\label{eqn:epsilonbound}
\end{align}
for all $t\leq T$.

Copying \eqn{sumboth1}, we have 
\begin{align}
n^2&\geq \sum_{t=0}^{T-1} \frac{2}{L}\tr[ H \rhobig_t].\label{eqn:sumboth2}
\end{align}
Using \eqn{epsilonbound}, we have
\begin{align}
n^2>\frac{(1-p)2\epsilon T}{L}.
\end{align}
Setting $T \geq \frac{n^2L}{2(1-p)\epsilon}$ gives a contradiction. 
Therefore, for $T \geq
\frac{n^2L}{2(1-p)\epsilon}$, we must have $\tr[\Pi_0 \rhobig_T ]\geq p.$
\end{myproof}


\subsection{Difficulties with General Clauses}
\label{sec:dificultys}
We have only been able to prove the {\it Quantum Algorithm} solves Quantum $2$-SAT
in polynomial time when we restrict the form of the clauses. In this
section, we describe what breaks down when more general clauses
are included in the instance. In this section, we assume that for
YES instances, the solution is a product of single-qubit states.
(The instance can be easily pre-processed to deal with any two-qubit
product states in the solution, as in \cite{dBG15}.)
In the YES case, we consider a basis in which the satisfying assignment
takes the form $\ket{0}^{\otimes n}$, so in this basis clauses are of the form: 

\noindent{\bf General Clauses:} 
\be \label{eqn:general2} 
\Phi_\alpha=\proj{\phi_\alpha}, \hspace{.2cm} \textrm{ with } 
\ket{\phi_\alpha}=a_\alpha\ket{01}_{i,j}+
b_\alpha\ket{10}_{i,j}+c_\alpha\ket{11}_{i,j}, 
\ee
 
 \noindent The restricted clauses never cause 
 the expectation of $\oZ^2$ to decrease. 
 However, when we include General Clauses
the expectation of $\oZ^2$ can either increase or decrease 
under the action of $\cT_\alpha$, depending on the state of the system.

Consider a clause of the form $\Phi_\alpha=\proj{\phi_\alpha}$ with
$\ket{\phi_\alpha}=\ket{+1}_{1,2}$
acting on the state $\rhobig=\proj{011}_{1,2,3}$. (Here $\ket{+}$ is
the eigenvector of the $\sigmabig^x$ operator with eigenvalue 1.) One can easily check that
\begin{align}
\tr[\oZ^2\rhobig]=5,\hspace{1cm}\tr[\oZ^2\cT_\alpha(\rhobig)]=4.5,
\end{align}
so the expectation value of $\oZ^2$ decreases.

When there are sufficiently many General Clauses, but still
with a planted product state solution, $\ket{0}^{\otimes
n}$ is the only satisfying state, so one might guess that a good
tracking measure would be the expectation value of $\oZ$, which if it
always increases, would bring the system closer and closer to
$\ket{0}^{\otimes n}$. However, for General Clauses, $\oZ$ can also
increase or decrease, and in fact for $\rhobig$ and $\Phi_\alpha$ as
above,
\begin{align}
\tr[\oZ\rhobig]=2,\hspace{1cm}\tr[\oZ\cT_\alpha(\rhobig)]=1.75.
\end{align}

While in principle the expectation value $\oZ$ and $\oZ^2$ under the
action of $\cT$ can increase or decrease, in numerical experiments, we
find that they always increase.

The analysis in \sec{analysis} was simple because the changes in
expectation value of $\oZ$ and $\oZ^2$ did not depend on the details of the
state of the system, but rather only on the overlap of the state with
the satisfying subspace. With general clauses, the
changes in expectation value of $\oZ$ and $\oZ^2$ depend on the
specifics of the state of the system, making these operators less
useful as tracking devices. 

 
 \section{Conclusions}


We study a quantum generalization of Sch\"oning's algorithm. We show
this quantum algorithm can be used to solve Quantum SAT problems. In
particular, we show that it can solve, in polynomial time, Quantum
$2$-SAT with certain restrictions on the clauses.   It is possible
that this quantum algorithm succeeds in polynomial time for  Quantum
$2$-SAT without any restriction on the clauses, but we were not able
to show it.  Inspired by the classical analysis, we track quantities
like the total spin rather than energy. Furthermore, if the
Hamiltonian is also polynomially gapped, the algorithm will produce,
in polynomial time, a state that has high overlap with a satisfying
assignment.

There are many open questions related to this work. Is there
a way to extend our analysis to unrestricted Quantum 2-SAT? 
How does the algorithm perform on Quantum $k$-SAT for $k>2$?
Can the runtime bounds of our algorithm
can be improved?

\section{Acknowledgments} 
We thank Sevag Gharibian for helpful
discussions. EF was funded by NSF grant CCF-1218176 and ARO contract
W911NF-12-1-0486. SK acknowledges support from the
Department of Physics at MIT, iQuISE Igert  grant contract number DGE-0801525,
and the Department of
Defense. KT acknowledges the support from the Erwin Schr{\"o}dinger
fellowship, Austrian Science Fund (FWF): J 3219-N16 and funding
provided by the Institute for Quantum Information and Matter, an NSF
Physics Frontiers Center (NFS Grants PHY-1125565 and PHY-0803371) with
support of the Gordon and Betty Moore Foundation (GBMF-12500028).


\appendix
\section{Analysis with an Extended Clause Set}
\label{app:extended-clause}

In \sec{analysis}, we showed that the {\it Quantum Algorithm} can decide
Quantum 2-SAT if (in the YES case) the clauses are of a certain form, which we now call
Type I Clauses:\\

\noindent{\bf Type I Clauses:} 
\be \label{eqn:type1} 
\Phi_\alpha=\proj{\phi_\alpha}, \hspace{.2cm} \textrm{ with } 
\ket{\phi_\alpha}=a_\alpha\ket{01}_{i,j}+b_\alpha\ket{10}_{i,j}.
\ee
In this appendix, we will show that the {\it Quantum Algorithm} almost
matches the performance demonstrated in the main body of this paper,
when the restricted clause set is enlarged to include both 
Type I and Type II clauses:\\

\noindent{\bf Type II Clauses:} 
\be\label{eqn:type2} 
\Phi_\alpha=\proj{\phi_\alpha}, \hspace{.2cm} \textrm{ with } 
\ket{\phi_\alpha}=\ket{11}_{i,j}.
\ee

\noindent When all clauses are Type I or Type II, $\ket{0}^{\otimes n}$ is a
satisfying state.

In \sec{analysis} we showed that for $\Phi_\alpha$ a Type I clause,
\begin{align}
&\tr[\oZ\cT_\alpha(\rhobig)]-\tr[\oZ\rhobig]=0, \label{eqn:Z1_change}\\
&\tr[\oZ^2\cT_\alpha(\rhobig)]-\tr[\oZ^2\rhobig]=
2\tr[\Phi_{\alpha}\rhobig].\label{eqn:Z2_changeb}
\end{align}
We observe that Type II clauses exhibit the following properties:
\begin{align}\label{eqn:typeII-props}
\Phi_{\alpha} (\sigmabig_{i}^z+\sigmabig_{j}^z) =(\sigmabig_{i}^z+\sigmabig_{j}^z)\Phi_{\alpha} =-2\Phi_{\alpha},\hspace{1cm}
\Phi_{\alpha} \sigmabig_{i}^z\sigmabig_{j}^z=\sigmabig_{i}^z\sigmabig_{j}^z\Phi_{\alpha}=\Phi_{\alpha}.
\end{align}
Applying \eqn{typeII-props} to \eqn{expand-oZ} and to the analysis
in \eqrange{S2divide}{S2final}, we have that for Type II clauses
\begin{align}
\cT^\dagger_\alpha(\oZ)&=\oZ+\Phi_{\alpha},\label{eqn:Z1_change2}\\
\cT^\dagger_\alpha(\oZ^2)&=\oZ^2-2\Phi_{\alpha}+2\sum_{k\neq
i,j}\sigmabig_k^z\Phi_{\alpha}\label{eqn:Z2_change2}.
\end{align}

Combining the effects of Type I and Type II clauses, we have
\begin{align}
\cT^\dagger(\oZ)&=\oZ+\frac{1}{L}\sum_{\alpha\in \textrm{Type II}}\Phi_{\alpha},\label{eqn:Stype2}\\
\cT^\dagger(\oZ^2)&=\oZ^2+\frac{2}{L}\sum_{\alpha\in \textrm{Type I}}\Phi_{\alpha}
+\frac{2}{L}\sum_{\alpha\in \textrm{Type II}}\Big({-}\Phi_{\alpha}+\sum_{k\neq
i,j}\sigmabig_k^z\Phi_{\alpha}\Big)\label{eqn:S2type2}.
\end{align}

When only Type I clauses were present,
 the expectation of $\oZ^2$ could only increase, but
now Type II clauses can cause $\oZ^2$ to decrease. However, whenever
$\rhobig_t$ is not annihilated by all of the clauses, either the expectation
value of $\oZ$ increases (if a Type II clause is measured), or the expectation
value of $\oZ^2$ increases (if a Type I clause is measured). 
We show that in combination, these effects allow us to prove the following result.

\begin{resul}\label{res:both}
Given clauses $\{\Phi_{\alpha}\}$, where $\Phi_{\alpha}$ are Type I or
Type II,
\begin{align}\label{eqn:sumboth}
5n^2\geq\frac{2}{L}\sum_{t=0}^{T-1}
\emph{tr}[H\rhobig_t].
\end{align}
\end{resul}

We first discuss the consequences of \res{both}, and then give the proof.
Note that \eqn{sumboth} is almost identical to  \eqn{sumboth1} and \eqn{sumboth2}. 
The only difference is the factor of $5$ that appears on the left side of \eqn{sumboth}. Thus to
determine what happens when, in the YES case, we restrict to Type I and Type II
clauses, we  need only replace
  \eqn{sumboth1} and \eqn{sumboth2} by \eqn{sumboth}.

 In \res{T3overlap} the number of time steps needed increases by a factor of $5$ to obtain the same
 outcome.
In \res{T1overlap} we use the following transformation, which preserves the statement of the result:
\begin{align}
T&\rightarrow \frac{5f^2L^2n^2}{2},\no
N&\rightarrow T\left(\frac{fL-1}{fL}\right)^3-2fLn.
\end{align}
 Using this transformation in \res{T2overlap}, the outcome is identical when
we choose $f=\max\{22/(5c),1\}$.

We now proof \res{both}:

\begin{proof}
Since $n^2 \geq \tr[\oZ^2\rhobig]\geq 0$, for any state $\rhobig$,
\begin{align}
n^2&\geq \tr[\oZ^2\rhobig_T ]- \tr[\oZ^2\rhobig_0]\no
=&\sum_{t=0}^{T-1}\left( \tr[\oZ^2\rhobig_{t+1}]-\tr[\oZ^2\rhobig_{t}] \right)\no
=&\sum_{t=0}^{T-1}\frac{2}{L}
\left(\sum_{\alpha \in \textrm{Type I}}\tr[\Phi_{\alpha}\rhobig_t]+\sum_{ \alpha \in \textrm{Type II}}
\tr\Big[\Big({-}1+\sum_{k\neq i, j}\sigmabig_k^z\Big)\Phi_{\alpha}\rhobig_t\Big]\right),
\end{align}
where we have used \eqn{S2type2} in the last line.

In the Type II sum, the term $({-}1+\sum_{k\neq i,j}\sigmabig_k^z)$
has eigenvalues that are larger than $-(n-1)$, so using that
$\Phi_\alpha$ and $({-}1+\sum_{k\neq i,j}\sigmabig_k^z)$ commute (they act
on different qubits), we
obtain
\begin{align}\label{eqn:step2} 
n^2\geq\frac{2}{L}\sum_{t=0}^{T-1}
\left(\sum_{\alpha \in \textrm{Type I}}\tr[\Phi_{\alpha}\rhobig_t]-(n-1)\sum_{\alpha \in \textrm{Type II}}\tr[\Phi_{\alpha}\rhobig_t]\right)
\end{align}

We have
\begin{align}
\tr[H\rhobig_t]=\sum_{\alpha \in \textrm{Type II}}\tr[\Phi_{\alpha}\rhobig_t]+
\sum_{\alpha \in \textrm{Type I}}\tr[\Phi_{\alpha}\rhobig_t],
\end{align}
which we can plug into \eqn{step2} to obtain
\begin{align}\label{eqn:step3} 
n^2\geq\frac{2}{L}\sum_{t=0}^{T-1}
\left(\tr[H\rhobig_t]-n\sum_{\alpha \in \textrm{Type II}}\tr[\Phi_{\alpha}\rhobig_t]\right).
\end{align}

We now bound the term involving the Type II clauses. From  \eqn{Stype2} we have
\begin{align}\label{eqn:sbound}
\sum_{t=0}^{T-1}\frac{1}{L}\sum_{\alpha\in \textrm{Type II}}
\tr[\Phi_{\alpha}\rhobig_t]&=
\sum_{t=0}^{T-1}\left(\tr[\oZ\rhobig_{t+1}]-\tr[\oZ\rhobig_t]\right)\no
&=\tr[\oZ\rhobig_T]-\tr[\oZ\rhobig_0]\no
&\leq 2n,
\end{align}
where in the last line we have used that for any $\rhobig$, we have $-n\leq \tr[\oZ\rhobig]\leq n$.
Plugging \eqn{sbound} into \eqn{step3}, we have 
\begin{align}\label{eqn:step4} 
5n^2\geq\frac{2}{L}\sum_{t=0}^{T-1}
\tr[H\rhobig_t].
\end{align}
\qedsymbol
\end{proof}

\bibliographystyle{unsrt}
\bibliography{Q-Schoening}
\end{document}